\documentclass[letterpaper, 10 pt, conference]{ieeeconf}  % Comment this line out if you need a4paper

\IEEEoverridecommandlockouts                              % This command is only needed if 
\usepackage{graphicx} % for pdf, bitmapped graphics files
\usepackage{epsfig} % for postscript graphics files
\usepackage{amsmath} % assumes amsmath package installed
\usepackage{amssymb,mathrsfs}  % assumes amsmath package installed
\usepackage{theorem}
\newtheorem{definition}{Definition}
\newtheorem{remark}{Remark}

\usepackage[breaklinks]{hyperref}

\hypersetup{pdfborder={0 0 0},
            breaklinks=true,%
            bookmarks=true,%
            bookmarksnumbered=true,%
            colorlinks=true,%
            citecolor=red,%
            linkcolor=blue,%
}

 \newtheorem{theorem}{Theorem}[section]

\newtheorem{Lemma}[theorem]{Lemma}

 \newtheorem{problem}[theorem]{Problem}
\usepackage[vlined,ruled,linesnumbered]{algorithm2e} %for algorithms
\def\1{\mathbf{1}}

\newcommand{\subscr}[2]{{#1}_{\textup{#2}}}

\newcommand{\abs}[1]{|{#1}|}

\usepackage[vlined,ruled,linesnumbered]{algorithm2e} %for algorithms

\title{\LARGE \bf
Cooperative Evasion by Translating Targets with Variable Speeds
}

\author{Shivam Bajaj \and Eloy Garcia \and Shaunak D. Bopardikar
\thanks{Shivam Bajaj and Shaunak D. Bopardikar are with the Department of Electrical and Computer Engineering, Michigan State University. email:
     \texttt{bajajshi@msu.edu, shaunak@msu.edu}. Eloy Garcia is with the Wright-Patterson Air Force Research Laboratory, Dayton OH. email: \texttt{eloy.garcia.2@us.af.mil}}}

\begin{document}

\maketitle
\thispagestyle{empty}
\pagestyle{empty}

%%%%%%%%%%%%%%%%%%%%%%%%%%%%%%%%%%%%%%%%%%%%%%%%%%%%%%%%%%%%%%%%%%%%%%%%%%%%%%%%
\begin{abstract}
We consider a problem of cooperative evasion between a single pursuer and multiple evaders in which the evaders are constrained to move in the positive $Y$ direction. The evaders are slower than the vehicle and can choose their speeds from a bounded interval. The pursuer aims to intercept all evaders in a given sequence by executing a Manhattan pursuit strategy of moving parallel to the $X$ axis, followed by moving parallel to the $Y$ axis. The aim of the evaders is to cooperatively pick their individual speeds so that the total time to intercept all evaders is maximized. We first obtain conditions under which evaders should cooperate in order to maximize the total time to intercept as opposed to each moving greedily to optimize its own intercept time. Then, we propose and analyze an algorithm that assigns evasive strategies to the evaders in two iterations as opposed to performing an exponential search over the choice of evader speeds. We also characterize a fundamental limit on the total time taken by the pursuer to capture all evaders when the number of evaders is large. Finally, we provide numerical comparisons against random sampling heuristics. 
\end{abstract}
\section{Introduction}
We consider a single pursuer multi-evader pursuit evasion problem in which the aim of the pursuer is to intercept all of the evaders in a fixed given sequence. The evaders are constrained to move along the positive $Y$ direction. The pursuer follows the Manhattan distance, i.e., moving parallel to the $X$-axis followed by moving parallel to the $Y$-axis. The aim of the evaders is to cooperatively maximize the total time to intercept all evaders. Such a set-up arises in riot control or border protection scenarios in which a ground or air vehicle would like to optimally visit mobile locations headed toward a boundary/asset, or in UAV monitoring of vehicles along a highway. This setup is also applicable in multiple robotic decoy deployment~\cite{ragesh2014analysis}.
\subsection{Related work}
Since the seminal work by Isaacs in \cite{isaacs1999differential}, much has been done in the field of pursuit evasion with a lot of focus on multi-agent pursuit evasion \cite{makkapati2019optimal,girard2015proportional,selvakumar2016evasion}. 
The case of a single pursuer and 2 evaders has been extensively analyzed \cite{fuchs2010cooperative}, \cite{zemskov1997construction}. Protector-Prey-Predator \cite{oyler2016pursuit} and Target-Attacker-Defender differential game \cite{garcia2014cooperative} are some examples of this scenario. With more than two evaders, the complexity of the problem grows exponentially with number of evaders. The problem of successive pursuit with cooperative multiple evaders is considered in \cite{chikrii1987pursuit, shevchenko2008guaranteed, belousov2010solving} and \cite{liu2013evasion}. Our problem differs from \cite{Scott2018OptimalES,krishnamoorthy2013optimal} as the pursuer follows a fixed strategy and the evaders are constrained to move in a fixed direction and can choose their individual speeds from a bounded interval to maximize the total intercept time. Thus, the evasive strategies are based on the range of evader speeds. 
% \nocite{Scott2018OptimalES,krishnamoorthy2013optimal}
% ,alexopoulos2015cooperative,kalyanam2013optimal

\subsection{Contributions}
We consider an optimal evasion problem between a single pursuer and $n$ evaders. The pursuer moves with unit speed. The evaders are constrained to move in the positive $Y$ direction such that their speeds $v_i, i\in \{1, \dots, n\}$, lie in the interval $[\subscr{u}{min},\subscr{u}{max}]$ with $0<\subscr{u}{min}<\subscr{u}{max}<1$. The evaders need to choose their speeds in order to maximize the total intercept time. We first present a complete solution to the optimal evasion problem for $n \leq 2$. We then show, for general $n$, that the optimal choice of the speed for each evader is one of the extremes, i.e., $\subscr{u}{min}$ or $\subscr{u}{max}$. We further show that, by enforcing cooperation among evaders, they are able to maximize the total intercept time. In order to implement the cooperative strategies, it is important to determine the conditions under which cooperation is optimal. Such conditions are also provided in this paper. We present an algorithm which assigns the evasive strategies to the evaders in two iterations as opposed to performing an exponential search over the choice of evader speeds. For sufficiently large $n$, for which the global optimum is difficult to compute, we establish a fundamental upper bound to the total intercept time taken by the pursuer to capture all evaders. Finally, we provide comparisons through numerical results.
\subsection{Organization}

The paper is organized as follows. Section \ref{sec:Prob} comprises the formal problem definition. In section \ref{sec:Strategy}, we derive an evasive strategy for multiple evaders and provide a Sequential-Greedy-Cooperation algorithm. Section \ref{sec:fundamental} establishes a fundamental upper bound on the total time to intercept all evaders. Section \ref{sec:Sims} presents the numerical simulations. Finally, section \ref{sec:Conc} summarizes this paper and outlines directions for future work.  

% Cooperation withing different agents in a multi-agent setting is a major challenge. A joint pursuit of a single evader is considered in \cite{huang2011guaranteed}. The authors propose a control scheme based on partitioning of the game domain by using Voronoi partition and show that the capture is guaranteed and independent of evaders actions. \cite{liu2013evasion} and \cite{Scott2018OptimalES} consider multi-agent pursuit evasion with motion constraints. Some Target-Attacker-Differential games also consider cooperation \cite{garcia2014cooperative}. In \cite{garcia2019cooperative}, the authors consider a zero sum differential game between two pursuers and a single evader in which the evader strives to be as close as possible to a goal line whereas the pursuers cooperatively try to capture the evader as far as possible from the same.

\section{Problem formulation}\label{sec:Prob}

We consider an optimal evasion problem played between a single pursuer with simple motion and $n$ mobile evaders. We denote the pursuer as $P$ and evaders as $E_i$ with $i \in \{1, \dots, n\}$. The pursuer with initial location at $(X,Y)$ is assumed to be moving with unit speed either along the $X$ or the $Y$ axis. We term this pursuit strategy as \emph{Manhattan pursuit}, and is formally defined as follows.
\begin{definition}[Simple Manhattan pursuit] Given initial locations $(x_i,y_i)$ and $(X,Y)$ of an evader $E_i$ and the pursuer $P$ respectively, the pursuer 
\begin{enumerate}
\item moves with unit speed along the positive or negative $X$ direction until $X(t) = x_i$ and then,
\item moves with unit speed along positive or negative $Y$ axis to intercept the evader.
\end{enumerate}
\end{definition}
The evaders, initially located at $\{(x_1, y_1), \dots, (x_n, y_n)\}$, are constrained to move along the positive $Y$ direction with simple motion such that their instantaneous speeds $v_i$, $i \in \{1, \dots, n\}$, lie in the interval $[\subscr{u}{min}, \subscr{u}{max}]$ with $0<\subscr{u}{min}<\subscr{u}{max}<1$ (Fig.  \ref{fig:problem_setup}). The pursuer is said to $\emph{intercept}$ the $i^{th}$ target when its location coincides with that of the $i^{th}$ target. The game $\emph{terminates}$ when the pursuer intercepts the last evader. A \emph{strategy} for an evader $E_i$ is a measurable function, defined as $v_i(\{X(t),Y(t)\}, \{x_i(t),y_i(t)\}_{i=1}^n)\to [\subscr{u}{min},\subscr{u}{max}]$, where the notation $\{X(t), Y(t)\}$ denotes the set of all locations $(X(\tau), Y(\tau)), \forall \tau \in [0,t]$. 
\begin{figure}[t]
    \centering
    \includegraphics[scale=0.45]{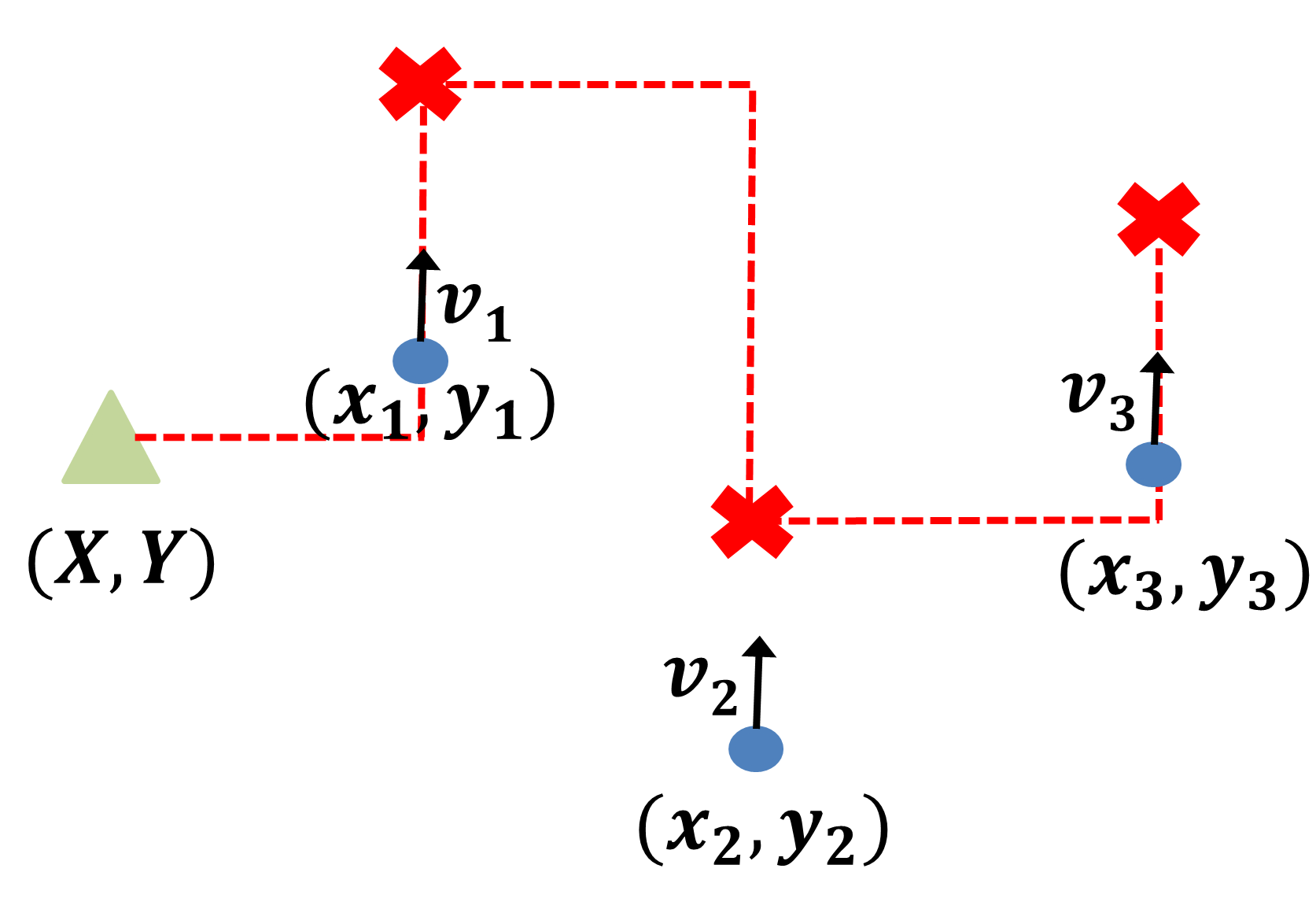}
    \caption{\small{Problem setup. The triangle represents the pursuer and the blue dots represent the evaders. The red dashed line represents the path taken by the pursuer to intercept the evaders. The cross represents the intercept locations.}}
    \label{fig:problem_setup}
\end{figure}
 The goal is to solve the following problem.
\begin{problem}[Optimal evasion] \label{prob:fixedorder} Given that the pursuer follows a fixed order to intercept the evaders, determine strategies $v_1^*, v_2^*, \dots, v_n^*$ for the evaders that maximizes the total time $T_n$ taken by the pursuer to intercept all $n$ evaders. \end{problem}

\section{Evasive Strategy}\label{sec:Strategy}
We begin with the case of a single evader followed by two evaders, and then present the more general case. We start by defining the following simple Manhattan pursuit strategy.

\subsection{Single Evader}
In this section, we first consider the case of a single evader and a pursuer, located at $(x_1,y_1)$ and $(X,Y)$ respectively. We first present a result on the time taken to intercept a single evader. This will be used in deriving the optimal strategy for the evader. We denote $\Delta x_{i}^{i-1}:= \abs{x_{i-1}-x_i}$, where $x_0=X$ and $T_{i}^{i-1}(v_i)$ as the time taken by the pursuer to intercept evader $i$ moving with speed $v_i$ after intercepting evader $i-1$. Specifically, $T_{1}^{0}(v)$ is the time taken to intercept the first evader moving with speed $v$. For brevity, we omit the proofs for the single evader case as they can be derived by following the steps in the proof of the general case presented later.

\begin{Lemma}[Time to intercept a single evader]\label{lem:intercepttimeE1}
% Given the initial locations $(x_1,y_1)$ and $(X,Y)$ of $E_1$ and $P$ respectively, 
The time $T_{1}^0$ taken by $P$ to intercept $E_1$ is
\begin{align*}
  T_{1}^0(v)= \begin{cases}
            \frac{\Delta x_{1}^0+ y_1-Y}{1-v}, \text{ if } \Delta x_{1}^0 >  (Y-y_1)/v,\\
            \frac{\Delta x_{1}^0+ Y-y_1}{1+v},\text{ otherwise.}
    \end{cases}
\end{align*}
\end{Lemma}
% \begin{proof}
% Omitted for brevity.
% % The outline of the proof is as follows.
% % Consider the case when $\Delta x_{1}^0 >  (Y-y_1)/v$. In other words, after stage (1) of the simple pursuit strategy, the $x$-coordinate of the pursuer equals $x_1$ and at the same time, the evader's $y$-coordinate exceeds $Y$. Therefore, $T_{1}^0(v) = \frac{\Delta x_{1}^0+ y_1-Y}{1-v}$.
% % The second case is derived analogously.
% \end{proof}

% \begin{definition}[Intercept point]\label{def:int_pt_MD_1} 
% The pursuer intercepts $E_1$ moving at constant speed of $v$ at the \emph{intercept point} $(x_1, \subscr{Y}{int})$ defined by
% % \begin{align*}
% % \subscr{Y}{int}(v) :=  y_1 + vT_{(0,1)}(v) &= \begin{cases} y_1 + v\frac{\Delta x_{(0,1)}+ y_1-Y}{1-v} = \frac{y_1 + v\Delta x_{(0,1)} - vY}{1-v}, \text{ if } \Delta x_{(0,1)} >  (Y-y_1)/v, \\
% % y_1 + v \frac{\Delta x_{(0,1)}+ Y-y_1}{1+v} = \frac{y_1 + v\Delta x_{(0,1)}+vY}{1+v}, \text{ otherwise.}\\
% % \end{cases}
% % \end{align*}
% \begin{align*}
% &\subscr{Y}{int}(v) := \begin{cases} \frac{y_1 + v\Delta x_{1}^0 - vY}{1-v}, \text{ if } \Delta x_{1}^0 >  \frac{Y-y_1}{v}, \\
%  \frac{y_1 + v\Delta x_{1}^0+vY}{1+v}, \text{ otherwise.}\\
% \end{cases}
% \end{align*}
% \end{definition}

\begin{Lemma}[Monotonicity of time to intercept]\label{lem:mon_1_E}
% Given initial locations $(x_1,y_1)$ and $(X,Y)$ of evader and pursuer respectively, then
The time $T_{1}^0(v)$ is a monotonically increasing function of $v$ if $\Delta x^0_1>\frac{Y-y_1}{v}$. Otherwise, $T_{1}^0(v)$ is a monotonically decreasing function of $v$.
% there exists a speed $v'=\frac{Y-y_1}{\Delta x^0_1} \in [\subscr{u}{min},\subscr{u}{max}]$ such that 
%  \begin{itemize}
%     \item If $v<\frac{Y-y_1}{\Delta x^0_1}$,  then $T_{1}^0(v)$ is a monotonically decreasing function of $v$,
%     \item Else, $T_{1}^0(v)$ is a monotonically increasing function of $v$.
% \end{itemize}
\end{Lemma}
% \begin{proof}
% Omitted for brevity.
% % Clearly, for $\Delta x^0_1>\frac{Y-y_1}{v}$, from Lemma \ref{lem:intercepttimeE1}, $\frac{dT_1^0(v)}{dv}>0$ and for $\Delta x^0_1\leq \frac{Y-y_1}{v}$, $\frac{dT_1^0(v)}{dv}<0$.
% \end{proof}

\begin{remark}
The time to intercept is monotonic even when the pursuer follows a Euclidean strategy, i.e., given the initial locations of $E_1$ and $P$, as $(x_1,y_1)$ and $(X,Y)$ respectively, the vehicle moves towards $(x_1,y_1+vT_1^0)$, where $T_1^0$ is
\begin{align*}
    \frac{(y_1-Y)v}{1-v^2}+\sqrt{\frac{(X-x_1)^2}{1-v^2}+\frac{(Y-y_1)^2}{(1-v^2)^2}}.
\end{align*}
% $\frac{(y_1-Y)v}{1-v^2}+\sqrt{\frac{(X-x_1)^2}{1-v^2}+\frac{(Y-y_1)^2}{(1-v^2)^2}}$.
\end{remark}
\medskip

Lemma \ref{lem:mon_1_E} characterizes the monotonic nature of $T_{1}^0(v)$. This only means that the maximum is achieved at one of the extremes. The next theorem characterizes the evader's optimal choice of speed.
%, i.e., if the evader moves with $\subscr{u}{min}$ or $\subscr{u}{max}$.
% \begin{figure}
%     \centering
%     \includegraphics[scale=0.3]{extreme_vale_MD.png}
%     \caption{\small{Single evader case. The blue dot represents the evader and the pursuer is represented by the triangle. The intercept locations $\subscr{I}{min}$, $\subscr{I}{max}$, $\subscr{I'}{min}$ and $\subscr{I'}{max}$ are represented by yellow and red cross respectively. (a) Intercept locations of $E_1$ when $(y_1=y')$. (b) Intercept locations of $E_1$ when $y_1<y'$ and (c) Intercept locations of $E_1$ when $y_1>y'$.}}
%     \label{fig:extreme_MD}
% \end{figure}

\begin{theorem}[Single evader optimal strategy]\label{lem:opt_MD_1}
Given the initial locations $(x_1,y_1)$ and $(X,Y)$ of the evader and the pursuer respectively, 
the optimal strategy $v^*$ for the evader is
\begin{align*}
v^* =\begin{cases}
     \subscr{u}{min}, \text{ if }y_1 < Y - \Delta x_{1}^0 (\frac{\subscr{u}{min}+\subscr{u}{max}}{2+\subscr{u}{min}-\subscr{u}{max}})\\
     \subscr{u}{max}, \text{ otherwise. }
\end{cases}
\end{align*}
\end{theorem}

\begin{proof}
We provide only an outline. We find a location $(x_1,y')$ such that $T_{1}^0(\subscr{u}{min})=T_{1}^0(\subscr{u}{max})$, where $T_{1}^0(\subscr{u}{min})$ (resp. $T_{1}^0(\subscr{u}{max})$) is the time to intercept when evader moves with $\subscr{u}{min}$ (resp. $\subscr{u}{max}$). From Lemma \ref{lem:intercepttimeE1}, $T_{1}^0(\subscr{u}{min})=T_{1}^0(\subscr{u}{max})\Rightarrow y'=Y-\Delta x_{1}^0\big(\frac{\subscr{u}{min}+\subscr{u}{max}}{2+\subscr{u}{min}-\subscr{u}{max}}\big)$.
% \begin{align*}
%   y'=Y-\Delta x_{1}^0\big(\frac{\subscr{u}{min}+\subscr{u}{max}}{2+\subscr{u}{min}-\subscr{u}{max}}\big).
% \end{align*}
This means that if $y_1=y'$, then from Lemma \ref{lem:mon_1_E}, it would not matter if the evader moves with $\subscr{u}{min}$ or $\subscr{u}{min}$ as $T_{1}^0(\subscr{u}{min})=T_{1}^0(\subscr{u}{max})$, and $T_1^0$ will be maximum at both $\subscr{u}{min}$ and $\subscr{u}{max}$. If $y_1<y'$, then, from Lemma \ref{lem:mon_1_E}, either $T_{1}^0(\subscr{u}{min})<T_{1}^0(\subscr{u}{max})$ or $T_{1}^0(\subscr{u}{min})>T_{1}^0(\subscr{u}{max})$. Assuming $T_{1}^0(\subscr{u}{min})<T_{1}^0(\subscr{u}{max}) \Rightarrow y_1>y'$ and thus, by contradiction, we get the result. The second case is analogous and we get the result.
\end{proof}
\medskip

We now consider the case of two evaders and derive the optimal evasion strategies for both evaders. We say that an evader $E_i$ moves \emph{greedy} if it moves with speed that maximizes its own intercept time. An evader \emph{cooperates} if it moves with a speed that maximizes the total intercept time. We denote the greedy strategy of evader $i$ as $v_{i_g}^*$ and the cooperative strategy as $v_{i_c}^*$.
\subsection{Two evaders}
Similar to previous section, we first derive an expression for the time taken to intercept the evaders followed by the optimal strategy for both evaders.

Let the first evader $E_1$ be located at $(x_1,y_1)$ and move with speed $v_1$ and the second evader $E_2$ be located at $(x_2,y_2)$ and move with speed $v_2$. Then, the following result summarizes the time to intercept $E_2$ after intercepting $E_1$. For ease of reference, we introduce the following condition:
\begin{equation}\label{eq:2evader-condition}
 \Delta x_{2}^1 > \frac{y_1-y_2 + (v_1-v_2)T_{1}^0(v_1)}{v_2}.
 \end{equation}

\begin{Lemma}[Time to intercept $E_2$]\label{lem:intercepttimeE2}
The time $T_{2}^1(v_1,v_2)$ taken by $P$ to intercept $E_2$ after intercepting $E_1$ is
\begin{align*}
    T_{2}^1(v_1,v_2) =
    \begin{cases}
        \frac{\Delta x_{2}^1 + y_2-y_1 + (v_2-v_1)T_{1}^0(v_1)}{1-v_2},\text{ if \eqref{eq:2evader-condition} holds,}\\
        \frac{\Delta x_{2}^1 + y_1-y_2 + (v_1-v_2)T_{1}^0(v_1)}{1+v_2},\text{ otherwise.}
    \end{cases}
\end{align*}
\end{Lemma}
\begin{proof}
Consider the case when condition \eqref{eq:2evader-condition} holds, after the intercept of $E_1$. This means that after the completion of stage (1) for the pursuit of $E_2$, the evader's $Y$-coordinate strictly exceeds $Y=y_1+v_1T_{1}^0(v_1)$. The additional time to intercept the second evader is
\[\frac{y_2+v_2T_{1}^0(v_1)-y_1-v_1T_{1}^0(v_1)+v_2\Delta x_{2}^1}{1-v_2}.\]
Thus, total time to intercept $E_2$ after intercepting $E_1$ is
\begin{align*}
&T_{2}^1(v_1,v_2)
% &\Delta x_{(1,2)}+\frac{y_2+(v_2-v_1)T_{(0,1)}(v_1)-y_1+v_2\Delta x_{(1,2)}}{1-v_2}\\
=\frac{\Delta x_{2}^1 + y_2-y_1 + (v_2-v_1)T_{1}^0(v_1)}{1-v_2}.
\end{align*}
The second case can be derived analogously and this concludes the proof.
\end{proof}

\begin{Lemma}[Monotonicity of time to intercept $E_2$]\label{lem:mon_2_E}
Given that $E_1$ moves with $v_1$, the time $T_{2}^1(v_1,v_2)$ is monotonically increasing function of $v_2$ if condition \eqref{eq:2evader-condition} holds. Otherwise, $T_{2}^1(v_1,v_2)$ is a monotonically decreasing function of $v_2$.
% there exists a speed $v'=\frac{Y-y_1}{\Delta x^0_1} \in [\subscr{u}{min},\subscr{u}{max}]$ such that 
%  \begin{itemize}
%     \item If $v<\frac{Y-y_1}{\Delta x^0_1}$,  then $T_{1}^0(v)$ is a monotonically decreasing function of $v$,
%     \item Else, $T_{1}^0(v)$ is a monotonically increasing function of $v$.
% \end{itemize}
\end{Lemma}
\begin{proof}
From Lemma \ref{lem:intercepttimeE2}, $\frac{dT_{2}^1(v_1,v_2)}{dv_2}>0$ if condition \eqref{eq:2evader-condition} holds and $\frac{dT_{2}^1(v_1,v_2)}{dv_2}<0$, otherwise. 
\end{proof}
\medskip
We now characterize an optimal greedy strategy for $E_2$. In what follows, we denote $V:=\frac{\subscr{u}{min}+\subscr{u}{max}}{2+\subscr{u}{min}-\subscr{u}{max}}$.

\begin{Lemma}[$E_2$'s greedy strategy]\label{lem:greed_2_MD}
The greedy strategy $v_{2g}^*$ for $E_2$ for a greedy $E_1$ moving with $v_{1g}^*$ is
\begin{align*}
     v_{2g}^*=
     \begin{cases}
        \subscr{u}{max}, \text{ if }
        y_2\geq y_1-\Delta x_{2}^1V+(v_{1g}^*-V)T_{1}^0(v_{1g}^*),\\
        \subscr{u}{min},\text{ otherwise.}
     \end{cases}
\end{align*}
% \begin{itemize}
%     \item If $y_2\geq y_1-\Delta x_{(1,2)}\big(\frac{\subscr{u}{min}+\subscr{u}{max}}{2+\subscr{u}{min}-\subscr{u}{max}}\big)+(v_{1g}^*-\frac{\subscr{u}{min}+\subscr{u}{max}}{2+\subscr{u}{min}-\subscr{u}{max}})T_{(0,1)}$, then,
%     \[
%         v_2^*=\subscr{u}{max}
%     \]
%     \item else
%     \[
%         v_2^*=\subscr{u}{min}
%     \]
% \end{itemize}
%where $T_{(0,1)}(v_{1g}^*)$ is the time to intercept $E_1$.
\end{Lemma}
\begin{proof}
From Lemma \ref{lem:mon_2_E}, $T_{2}^1(v_{1g}^*,v_2)$ 
is maximized at either $v_2=\subscr{u}{min}$ or $v_2=\subscr{u}{max}$. The aim is to find the critical location $y_2'$ such that if $E_2$ was located at $(x_2,y_2')$, then the time $T_{2}^1(v_{1g}^*,\subscr{u}{max})$= $T_{2}^1(v_{1g}^*,\subscr{u}{min})$. Note that this is possible only if condition \eqref{eq:2evader-condition} holds for $v_2=\subscr{u}{max}$ and does not hold for $v_2=\subscr{u}{min}$. From Lemma \ref{lem:intercepttimeE2}, we get $y_2'=y_1-\Delta x_{2}^1V+
    (v_{1g}^*-V)T_{1}^0(v_{1g}^*)$.
% From Remark \ref{rem:const_speed}, let us consider that $E_2$ is moving with speed $V$. Then we can find the location of $E_2$ when the pursuer intercepts $E_1$, i.e., after time $T_{1}^0(v_{1g}^*)$, the y-coordinate of $E_2$ is $y_2+ VT_{1}^0(v_{1g}^*)$. Now, if 
% \begin{align*}
%     % &y_2'+ VT_{1}^0(v_{1g}^*)<
%     % y_1+v_{1g}^*T_{1}^0(v_{1g}^*)-\Delta x_{2}^1V\\
%      y_2'=y_1-\Delta x_{2}^1V+
%     (v_{1g}^*-V)T_{1}^0(v_{1g}^*).
% \end{align*}
This means that if $y_2=y_2'$, then irrespective of $E_2$'s choice of $\subscr{u}{min}$ or $\subscr{u}{max}$, the time to intercept $E_2$ will be the same and from Lemma \ref{lem:mon_2_E}, the time to intercept $E_2$ will be maximum at both $\subscr{u}{min}$ and $\subscr{u}{max}$, given that $E_1$ moves greedy.
Now, consider that the initial location of $E_2$  is such that $y_2<y_2'$. From Lemma \ref{lem:mon_2_E}, the time to intercept will be maximized only at either $\subscr{u}{min}$ or $\subscr{u}{max}$ and so assume that $T_{2}^1$ is maximized at $\subscr{u}{max}$, i.e., $T_{2}^1(v_{1g}^*,\subscr{u}{max})>T_{2}^1(v_{1g}^*,\subscr{u}{min})$. This implies $y_2>y_2'$.
% \begin{align*}
%     &\frac{\Delta x_{1}^0+y_1-Y}{1-\subscr{u}{max}}>\frac{\Delta x_{1}^0+Y-y_1}{1+\subscr{u}{min}},\\
%     &\Rightarrow y_1>Y-\frac{\Delta x_{1}^0(\subscr{u}{min}+\subscr{u}{max})}{2+\subscr{u}{min}-\subscr{u}{max}}
%     \Rightarrow y_1>y'.
% \end{align*}
This is a contradiction as $y_2<y_2'$. This means that $T_{2}^1$ will be maximized if $E_2$ moves at $\subscr{u}{min}$. Similarly, when $y_2\geq y_2'$ it can be shown that $T_{2}^1$ will  be maximized if $E_2$ moves at $\subscr{u}{max}$ and has been omitted for brevity.
In the case when condition \eqref{eq:2evader-condition} does not hold for $v_2=\subscr{u}{max}$ or holds for $v_2=\subscr{u}{min}$, then it implies that $y_2<y_2'$ and $y_2\geq y_2'$ respectively. This concludes the proof.
\end{proof}
\medskip

Lemma \ref{lem:greed_2_MD} yields a \emph{greedy} strategy for $E_2$ when $E_1$ and $E_2$ both move greedily. However, it might be better for the evaders to cooperate to maximize the total intercept time. We now characterize the conditions on cooperation between the two evaders.

\begin{figure}
    \centering
    \includegraphics[scale=0.3]{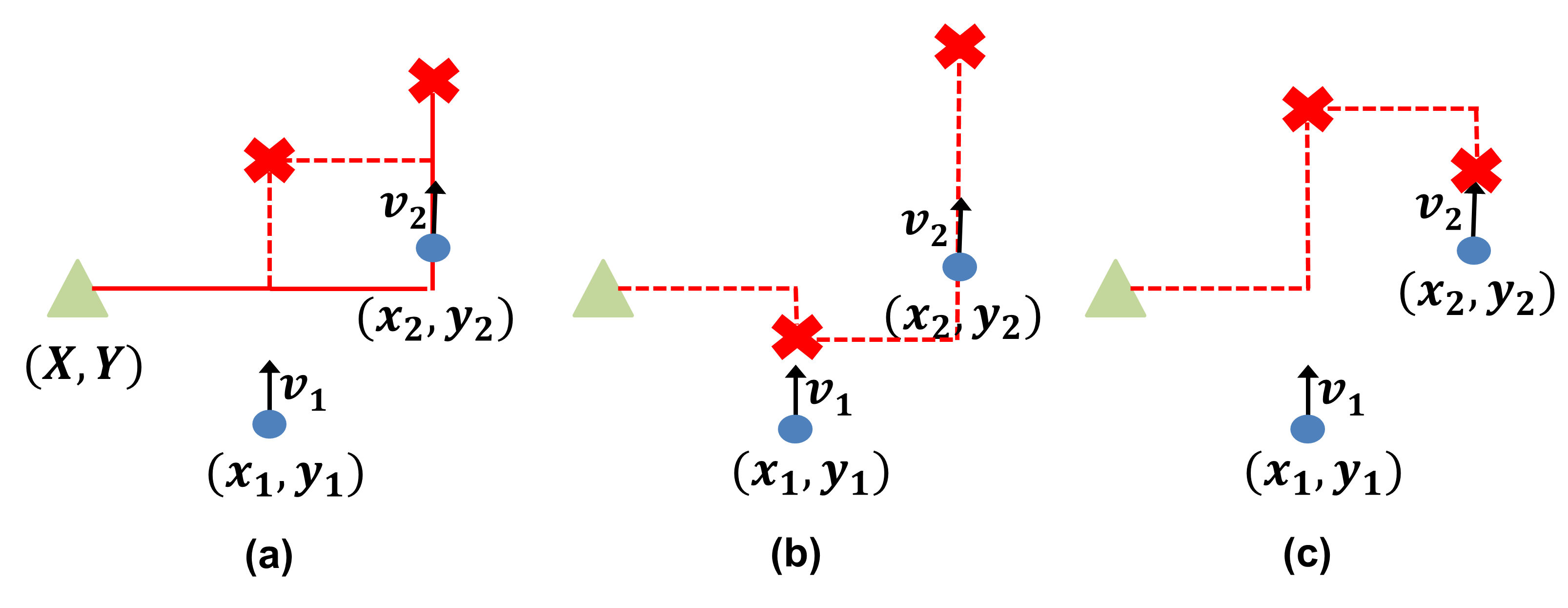}
    \caption{\small{Cases for cooperation. (a) The time (solid line) taken by the pursuer to intercept $E_2$ directly from $(X,Y)$ and time (dashed line) taken by the pursuer to intercept $E_1$ first then $E_2$ are the same. (b) $E_1$ and $E_2$ cooperate by moving at $\subscr{u}{min}$ and $\subscr{u}{max}$ respectively (c) $E_1$ and $E_2$ cooperate by moving at $\subscr{u}{max}$ and $\subscr{u}{min}$ respectively.}}
    \label{fig:MD_2_coop}
\end{figure}
We define that a point $A$, located at $(x_A,y_A)$, is \emph{above} point $B$, located at $(x_B,y_B)$, if $y_A>y_B$ and we define point $A$ is \emph{below} point $B$ if $y_A<y_B$.
\begin{Lemma}[Conditions on cooperation]\label{lem:coop_2E_MD}
Given the initial locations of $E_1$, $E_2$, and $P$ as $(x_1,y_1)$, $(x_2,y_2)$, and $(X,Y)$ respectively, $E_1$ cooperates with $E_2$ if 

\noindent {\bf (i) Case 1:}
\begin{align*}
            Y-\Delta x_{1}^0V\leq y_1\leq Y-\Delta x_{1}^0\subscr{u}{min}, \text{ and}
        \end{align*}
        \begin{equation}\label{eq:2_coop_MD}
        \begin{split}
            &y_2>y_1-\Delta x_{2}^1V+ \big(v_{1g}^*-V)T_{1}^0(v_{1g}^*) +\\ &2\big(\frac{\subscr{u}{min}\Delta x_{1}^0+y_1-Y}{2+\subscr{u}{min}-\subscr{u}{max}}\big),
            \end{split}
        \end{equation}
        
\noindent {\bf (ii) Case 2:}
\begin{align*}
            Y-\Delta x_{1}^0\subscr{u}{max}\leq y_1\leq Y-\Delta x_{1}^0V, \text{ and}
        \end{align*}
        \begin{equation}\label{eq:2_coop_MD_2}
        \begin{split}
            &y_2\leq y_1-\Delta x_{2}^1V+ (v_{1g}^*-V)T_{1}^0(v_{1g}^*) + \\
            &2\big(\frac{\subscr{u}{max}\Delta x_{1}^0+y_1-Y}{2+\subscr{u}{min}-\subscr{u}{max}}\big).
        \end{split}
        \end{equation}
%where $T_{(0,1)}(v_{1g}^*)$ is the intercept time of $E_1$.
\end{Lemma}
\begin{proof}
Let the initial location of $E_1$ satisfy $y_1\geq Y-\Delta x_1^0V$ and initial location of $E_2$ satisfy $y_2 \geq y_1-\Delta x_{2}^1V+(v_{1g}^*-V)T_{1}^0(v_{1g}^*)$. Then, from Theorem \ref{lem:opt_MD_1}, to maximize its own intercept time, $T_{1}^0(v_{1g}^*)$, $E_1$ moves with $\subscr{u}{max}$. Also, from Lemma \ref{lem:greed_2_MD}, $E_2$ moves with $\subscr{u}{max}$ in order to maximize its own intercept time $T_{2}^1(v_{1g}^*,v_2)$. As the pursuer follows the Manhattan pursuit strategy, the time taken to cover the path in the y-direction to intercept $E_1$ and then $E_2$ is the same as the time taken to cover the path in y-direction to intercept only $E_2$ from the initial pursuer location $(X,Y)$ (see Fig.\ref{fig:MD_2_coop} (a)). Mathematically,
% $T_{1}^0(\subscr{u}{max})+T_{2}^1(\subscr{u}{max},\subscr{u}{max})
%     =\frac{\Delta x_{1}^0+\Delta x_{2}^1+y_2-Y}{1-\subscr{u}{max}}$.
\begin{align*}
    T_{1}^0(\subscr{u}{max})+T_{2}^1(\subscr{u}{max},\subscr{u}{max})
    =\frac{\Delta x_{1}^0+\Delta x_{2}^1+y_2-Y}{1-\subscr{u}{max}}.
\end{align*}
Thus, in order to increase the total time to intercept, $E_1$ and $E_2$ need to cooperate. One way to cooperate is that $E_1$ moves with a speed $v_1$ such that $E_1$ is intercepted below the pursuer, i.e., $\Delta x_{1}^0<(Y-y_1)/v_1$ and $E_2$ moves greedily, i.e., with $\subscr{u}{max}$ (see Fig \ref{fig:MD_2_coop} (b)). This is only possible if 
% $\Delta x_{1}^0<(Y-y_1)/\subscr{u}{min}
    % \Rightarrow y_1<Y-\Delta x_{1}^0\subscr{u}{min}$.
\begin{align*}
    \Delta x_{1}^0<(Y-y_1)/\subscr{u}{min}
    &\Rightarrow y_1<Y-\Delta x_{1}^0\subscr{u}{min}.
\end{align*}
% By doing so, $E_1$ increases the additional time to intercept $E_2$. Mathematically, from Lemma \ref{lem:intercepttimeE1}, this additional time will be
% \begin{align*}
%     &2\times\frac{Y-y_1 - v\Delta x_{(0,1)}}{1+v}\\
% \end{align*}
% Thus, to maximize this additional time,
% \begin{align*}
%     &\operatorname{max}_{v_1} 2\times\frac{Y-y_1 - v\Delta x_{(0,1)}}{1+v}.\\
%     &\Rightarrow v_1^* = \subscr{u}{min}.
% \end{align*}

If $y_1>Y-\Delta x_{1}^0\subscr{u}{min}$, then for any speed $v_1 \in [\subscr{u}{min},\subscr{u}{max}]$ for $E_1$, $\Delta x_{1}^0>(Y-y_1)/v_1$, and so the total distance covered in the y-direction to intercept $E_1$ and $E_2$ will be the same as the total distance covered in the y-direction to intercept $E_2$, from the pursuer's initial location $(X,Y)$, irrespective of $E_1$'s choice. The only other way in which $E_1$ and $E_2$ can cooperate in order to increase the total intercept time is when $E_1$ moves greedily, i.e., with $\subscr{u}{max}$ and $E_2$ with speed $v_2$ such that $y_2 < y_1-\Delta x_{2}^1v_2+(v_{1g}^*-v_2)T_{1}^0(v_{1g}^*)$ such that intercept of $E_2$ is below the intercept location of $E_1$ (see Fig. \ref{fig:MD_2_coop} (c)). Thus, to determine which of the two scenarios yield the greater time to intercept, we arrive at a condition $T_{1}^0(\subscr{u}{min})+T_{2}^1(\subscr{u}{min},\subscr{u}{max}) > T_{1}^0(\subscr{u}{max}) + T_{2}^1(\subscr{u}{max},\subscr{u}{min})$.
% Mathematically, this distance is
% \[
%     2\times (\Delta x_{(1,2)}\subscr{u}{min}+y_1-Y).
% \]
%Now, let us assume that Equation \ref{eq:2_coop_MD} holds such that
% \begin{multline*}
%     T_{1}^0(\subscr{u}{min})+T_{2}^1(\subscr{u}{min},\subscr{u}{max}) > \\
%     T_{1}^0(\subscr{u}{max}) + T_{2}^1(\subscr{u}{max},\subscr{u}{min}).
% \end{multline*}
% Substituting the values yield equation \eqref{eq:2_coop_MD}.
\begin{multline*}
    \frac{\Delta x_{2}^1+y_2-y_1+(1-\subscr{u}{min})T_{1}^0(\subscr{u}{min})}{1-\subscr{u}{max}}>\\
    \frac{\Delta x_{2}^1+y_1-y_2+(1+\subscr{u}{max})T_{1}^0(\subscr{u}{max})}{1+\subscr{u}{min}}\\
    % \Rightarrow y_2>y_1-\Delta x_{2}^1V+\Delta x_{1}^0V+
    % (y_1-Y)\frac{(2+\subscr{u}{max}-\subscr{u}{min})}{2+\subscr{u}{min}-\subscr{u}{max}}\\
    \Rightarrow y_2>y_1-\Delta x_{2}^1V + (v_{1g}^*-V)T_{1}^0(v_{1g}^*)+\\
    \frac{2(\subscr{u}{min}\Delta x_{1}^0 + y_1-Y)}{2+\subscr{u}{min}-\subscr{u}{max}}.
\end{multline*}
Note that as $T_{1}^0(v_1)$ and $T_{2}^1(v_1,v_2)$ are monotonic in $v_1$ and $v_2$ respectively, so the above condition is checked only at the extreme values, i.e., $\subscr{u}{min}$ and $\subscr{u}{max}$.
Thus, this means that $E_1$ should cooperate only if equation \eqref{eq:2_coop_MD} holds. Furthermore, if the initial location of $E_1$ was such that $y_1\leq Y-\Delta x_{1}^0V$, then from Lemma \ref{lem:opt_MD_1}, $E_1$ moves greedily, i.e., with speed $\subscr{u}{min}$. Since, this will already ensure that $\Delta x_{1}^0<(Y-y_1)/\subscr{u}{min}$ and so, there is no need for cooperation. Case 2 is analogous. This concludes our proof.
\end{proof}
% \begin{remark}
% $E_1$ cooperates only when $y_1<Y$. Otherwise, $E_1$ moves greedy.
% \end{remark}
% \medskip

% The next result characterises an optimal cooperative strategy for $E_1$ based on the conditions in Lemma \ref{lem:coop_2E_MD}.
\begin{theorem}[Optimal cooperative strategy for $E_1$]\label{thm:opt_strat_2}
Given the initial locations $(x_1,y_1)$, $(x_2,y_2)$, and $(X,Y)$ of $E_1$, $E_2$ and $P$, respectively, if the conditions for cooperation in Lemma \ref{lem:coop_2E_MD} hold, then the optimal cooperative strategy $v_{1c}^*$ for evader $E_1$ is
\[
v_{1c}^*=
\begin{cases}
     %For case 1 from Lemma \ref{lem:coop_2E_MD}:
             \subscr{u}{min}, \text{ for case 1 from Lemma \ref{lem:coop_2E_MD} or,}\\
       \subscr{u}{max},  
     \text{ for case 2 from Lemma \ref{lem:coop_2E_MD}.}
\end{cases}
\]
\end{theorem}
\begin{proof}
Consider that case 1 of Lemma \ref{lem:coop_2E_MD} holds. Then
% \begin{align*}
%             Y-\frac{\Delta x_{(0,1)}(\subscr{u}{min}+\subscr{u}{max})}{2+\subscr{u}{min}-\subscr{u}{max}}\leq y_1\leq Y-\Delta x_{(0,1)}\subscr{u}{min},
%         \end{align*}
%         and 
%         \begin{align*}
%             y_2>y_1-\Delta x_{(1,2)}\frac{(\subscr{u}{min}+\subscr{u}{max})}{2+\subscr{u}{min}-\subscr{u}{max}}+ \big(v_{1g}^*-\frac{(\subscr{u}{min}+\subscr{u}{max})}{2+\subscr{u}{min}-\subscr{u}{max}}\big)T_{(0,1)}(v_{1g}^*) + 2(\frac{\subscr{u}{min}\Delta x_{(0,1)}+y_1-Y}{2+\subscr{u}{min}-\subscr{u}{max}}).
%         \end{align*}
 $E_1$ moves with speed $v_1$ such that $\Delta x_{1}^0<(Y-y_1)/v_1$. We know from Lemma \ref{lem:coop_2E_MD} that the conditions on cooperation ensure that the total time to intercept during cooperation is higher than the greedy choice. Since $T_{1}^0(v_1)$ is monotonic in $v_1$, from Lemma \ref{lem:mon_1_E}, $v_{1c}^*=\subscr{u}{min}$. The second case is derived analogously. This concludes the proof.
\end{proof}

In this subsection, we analyzed the case of 2 evaders, primarily to highlight the underlying problem structure.
Next, we will consider the case of $n$ evaders. Similar to the two evader case, we will first present a result on the time taken to intercept the $k^{th}$ evader after intercepting the $k-1^{th}$ evader. Then we will present results on the greedy and cooperative strategies between $E_k$ and $E_{k-1}$.

\subsection{$n$ Evaders}

For ease of presentation, we will denote $(y_i-y_j)$ as $\Delta y_{j}^i$ for some $i,j$ and for brevity, we denote $T_{i}^{i-1}(v_1,\dots,v_i)$ as $T_{i}^{i-1}(v_{-i},v_i)$.
We present the following condition for ease of reference.
\begin{equation}\label{eq:greed_n}
    \Delta x_{k}^{k-1}>\frac{\Delta y_{k}^{k-1} + (v_{k-1}-v_k)\sum_{i=1}^{k-1} T_{i}^{i-1}}{v_k}.
\end{equation}

\begin{Lemma}[Time to intercept $E_k$]\label{lem:intercepttime_En_MD}
The time $T_{k}^{k-1}(v_{-k},v_k)$ taken by $P$ to intercept $E_k$, moving with $v_k$, after intercepting $E_{k-1}$, moving with $v_{k-1}$, is
\begin{align*}
   T_{k}^{k-1} = \begin{cases}
        \frac{\Delta x_{k}^{k-1} + \Delta y_{k-1}^k + (v_k-v_{k-1})\sum_{i=1}^{k-1} T_{i}^{i-1}}{1-v_k},\text{ if \eqref{eq:greed_n} holds,}\\
        \frac{\Delta x_{k}^{k-1}+\Delta y_{k}^{k-1} + (v_{k-1}-v_k)\sum_{i=1}^{k-1} T_{i}^{i-1}}{1+v_k}, \text{otherwise.}
    \end{cases}
\end{align*}
\end{Lemma}

\begin{proof}
We establish this result using mathematical induction. Lemma \ref{lem:intercepttimeE2} yields the base of induction for $k=2$.
%Lemma \ref{lem:intercepttimeE2} presents the result for 2 evaders.
% we know that if $\Delta x_{(1,2)} > (y_1-y_2 + (v_1-v_2)T_{(0,1)})/v_2 $, then
%     \[
%         T_{(1,2)}(v_1,v_2) = \frac{\Delta x_{(1,2)} + y_2-y_1 + (v_2-v_1)T_{(0,1)}(v_1)}{1-v_2}
%     \]
%  else,
%     \[
%     T_{(1,2)}(v_1,v_2)=\frac{\Delta x_{(1,2)} + y_1-y_2 + (v_1-v_2)T_{(0,1)}(v_1)}{1+v_2}.
%     \]
Assume that for some $k=\bar k$, the result holds.
Consider that the initial location of the next evader, $E_{\bar k+1}$, is such that $\Delta x_{\bar k+1}^{\bar k} > (y_{\bar k}-y_{\bar k+1} + (v_{\bar k}-v_{\bar k+1})\sum_{i=1}^{\bar k} T_{i}^{i-1})/v_{\bar k+1}$. This means that after the completion of stage (1) of simple pursuit of $E_{\bar k+1}$, the X-coordinate of the pursuer equals $x_{\bar k+1}$ and at the same time, the evader's y-coordinate strictly exceeds the pursuers y-coordinate $(=y_{\bar k}+v_{\bar k}\sum_{i=1}^{\bar k}T_{i}^{i-1})$.
% Thus, the additional time to intercept the $E_{k+1}$ will be
%     \[
%         \frac{y_{k+1}-y_k+(v_{k+1}-v_{k})\sum_{i=1}^k T+v_{k+1}\Delta x_{(k,k+1)}}{1-v_{k+1}}.
%     \]
Thus, the time to intercept $E_{\bar k+1}$ after intercepting $E_{\bar k}$ will be 
% $\frac{\Delta x_{\bar k+1}^{\bar k} + y_{\bar k+1}-y_{\bar k} + (v_{\bar k+1}-v_{\bar k})\sum_{i=1}^{\bar k} T_{i}^{i-1}}{1-v_{\bar k+1}}$.
\begin{align*}
    %T_{(k,k+1)}
    %&=\Delta x_{(k,k+1)} + \frac{y_{k+1}-y_k+(v_{k+1}-v_{k})\sum_{i=1}^k T_{(i-1,i)}(v_1,\dots,v_i)+v_{k+1}\Delta x_{(k,k+1)}}{1-v_{k+1}}\\
    \frac{\Delta x_{\bar k+1}^{\bar k} + y_{\bar k+1}-y_{\bar k} + (v_{\bar k+1}-v_{\bar k})\sum_{i=1}^{\bar k} T_{i}^{i-1}}{1-v_{\bar k+1}}.
\end{align*}
Thus, by induction the result holds for any value of $\bar k$. The other case is derived analogously.
\end{proof}
\begin{Lemma}[Monotonicity of time to intercept]\label{lem:mon_n_E}
Given that each $E_i$, $i\in{1, \dots, k-1}$ moves with $v_i$, the time $T_{k}^{k-1}$ is monotonically increasing function of $v_k$ if condition \eqref{eq:greed_n} holds. Otherwise, $T_{k}^{k-1}$ is a monotonically decreasing function of $v_2$.
\end{Lemma}
\begin{proof}
We only provide an outline. We use induction to establish the result. Lemma \ref{lem:mon_2_E} yields the base of the induction. Assuming that the result holds for some $k=\bar{k}$, it can be checked that $\frac{dT_{\bar{k}+1}^{\bar{k}}}{dv_k}>0$ if condition \eqref{eq:greed_n} holds. Otherwise, $\frac{dT_{\bar{k}+1}^{\bar{k}}}{dv_k}<0$. This concludes the proof.
\end{proof}
\medskip

Since Lemma \ref{lem:mon_n_E} establishes that the time to intercept a $k^{th}$ evader is maximized at either $\subscr{u}{min}$ or $\subscr{u}{max}$, finding an optimal strategy for all evaders would require analyzing all $2^n$ possibilities in the worst case.

We now present an algorithm that assigns respective strategies to the evaders in just two iterations. 
% Algorithm \ref{algo:gesec} describes an algorithm that sequentially assigns respective strategies to the evaders.
% The Sequential-Greedy-Cooperation algorithm is defined in Algorithm \ref{algo:gesec}. 
The algorithm, summarized in Algorithm \ref{algo:gesec}, first assigns the greedy strategies to all evaders. Then, it assigns cooperative strategies by considering two sequentially paired evaders at a time. \\
Now, we will present the results that the algorithm uses in assigning the strategies.
\begin{algorithm}[t]
	\DontPrintSemicolon
	\SetAlgoLined
	Assign greedy speeds to all evaders\\
	\eIf{$E_i$ and $E_{i+1}$ can cooperate, $\forall \text{ }1<i<n$,}{
		Assign optimal cooperative strategy \;}
	{Assign optimal greedy strategy. \;
	}
	Repeat from step 2.\;
	\caption{Seq-GreC Algorithm}
	\label{algo:gesec}
\end{algorithm}

\begin{Lemma}[Evader k's greedy strategy]\label{lem:greed_n_MD}
The greedy strategy $v_{kg}^*$ for $E_k$, when each $E_i$, $i\in\{1, \dots, k-1\}$ moves with $v_{1g}^*$ is
\begin{align*}
    v_{kg}^*=
    \begin{cases}
        \subscr{u}{max},\text{ if } y_k\geq
        y_{k-1}-\Delta x_{k}^{k-1}V+\\(v_{(k-1)g}^*-V)\sum_{i=1}^{k-1}T_{i}^{i-1},\\
        \subscr{u}{min}, \text{ otherwise.}
    \end{cases}
\end{align*}
\end{Lemma}
\begin{proof}
Suppose the result holds for some $k=\bar k$.
% $y_k\geq y_{k-1}-\Delta x_{(k-1,k)}V+(v_{(k-1)g}^*-V)\sum_{i=1}^{k-1} T_{(i-1,i)}$ and so, $v_k^*=\subscr{u}{max}$. 
Consider the next evader, $E_{\bar k+1}$. Similar to the proof of Lemma \ref{lem:greed_2_MD}, we find $y_{\bar k+1}'=y_{\bar k}-\Delta x_{\bar k+1}^{\bar k}V+(v_{\bar kg}^*-V)\sum_{i=1}^{\bar k}T_{i}^{i-1}$. If $y_{\bar k+1}<y_{\bar k+1}'$, then, from Lemma \ref{lem:mon_n_E}, the time will be maximized at either $\subscr{u}{min}$ or $\subscr{u}{max}$. Thus, assuming $T_k^{k-1}(v_{1g}^*,\dots,v_{(k-1)g}^*,\subscr{u}{max})>T_k^{k-1}(v_{1g}^*,\dots,v_{(k-1)g}^*,\subscr{u}{min})$ yields $y_{\bar k+1}>y_{\bar k+1}'$ which is a contradiction and so $v_{kg}^*=\subscr{u}{min}$. Moreover, by induction, the result holds for any value of $\bar{k}$. Case 2 is proved analogously.
% According to Remark \ref{rem:const_speed}, we can consider that $E_{\bar k+1}$ moves with speed $V$. Then, after the intercept of $E_{\bar k}$, the y-coordinate of $E_{\bar k+1}$ will be located at
% $
% y_{\bar k} + V\sum_{i=1}^{\bar k}T_{i}^{i-1}.
% $
% If 
% %\begin{align*}
%     % &y_{k+1}+ V\sum_{i=1}^{k}T_{(i-1,i)}<y_k+v_{kg}^*\sum_{i=1}^{k}T_{(i-1,i)}-\Delta x_{(k,k+1)}V\\
%     $y_{\bar k+1}<y_{\bar k}-\Delta x_{\bar k+1}^{\bar k}V+(v_{\bar kg}^*-V)\sum_{i=1}^{\bar k}T_{i}^{i-1}$
% %\end{align*}
% then, from Remark \ref{rem:const_speed}, $v_{k+1}^*=\subscr{u}{min}$. 
% % Similarly, if 
% % \begin{align*}
% %     % &y_{k+1}+ \frac{\subscr{u}{max}+\subscr{u}{min}}{2+\subscr{u}{min}-\subscr{u}{max}}\sum_{i=1}^{k}T_{(i-1,i)}(v_1,\dots,v_i)\geq y_k+v_{kg}^*\sum_{i=1}^{k}T_{(i-1,i)}(v_1,\dots,v_i)-\Delta x_{(k,k+1)}\big(\frac{\subscr{u}{max}+\subscr{u}{min}}{2+\subscr{u}{min}-\subscr{u}{max}}\big)\\
% %     &y_{k+1}\geq y_k-\Delta x_{(k,k+1)}V+\big(v_{kg}^*-V )\sum_{i=1}^{k}T_{(i-1,i)}
% % \end{align*}
% % then, $v_{k+1}^*=\subscr{u}{max}$.
This concludes our proof.
\end{proof}
%Note that this result holds for any choice of speed for $E_1,\dots,E_{n-1}$. Since the Algorithm \ref{algo:gesec} requires the greedy choice of all the evaders, we consider that all the evaders are moving greedy in Lemma \ref{lem:greed_n_MD}.\\
\medskip

The previous lemma presented a result on the greedy strategy of any evader $E_k$. This result is the first step of the Algorithm \ref{algo:gesec}. As the second step of Algorithm \ref{algo:gesec} requires to check the conditions of cooperation between two consecutive evaders, we will now present a result on the conditions if two evaders should cooperate or not. We introduce the notation, $U:=\frac{2}{2+\subscr{u}{min}-\subscr{u}{max}}$.

\begin{Lemma}[Cooperation conditions for $E_{k-1}$]\label{lem:coop_nE_MD}
Given the initial locations of $E_{k-1}$, $E_k$, and $P$ as $(x_{k-1},y_{k-1})$, $(x_k,y_k)$, and $(X,Y)$ respectively, then $E_{k-1}$ will cooperate with $E_k$ if\\
\noindent {\bf (i) Case 1:}
        \begin{align*}
            &y_{k-2}+(v_{(k-2)a}^*-V)\sum_{i=1}^{k-2}T_{i}^{i-1}-\Delta x_{k-1}^{k-2}V\leq y_{k-1}\\
            &\leq y_{k-2}
            +(v_{(k-2)a}^*-\subscr{u}{min})\sum_{i=1}^{k-2}T_{i}^{i-1}- \Delta x_{k-1}^{k-2}\subscr{u}{min}
        \end{align*}
        and 
        \begin{align}\label{eq:n_coop_MD}
            &\Delta y_{k}^{k-1}>-\Delta x_{k-1}^k V+ (v_{(k-1) g}^*-V)\sum_{i=1}^{k}T_{i}^{i-1}+\nonumber\\
            &U\big(\subscr{u}{min}\Delta x_{k-1}^{k-2}+ \Delta y_{k-2}^{k-1} -(v_{(k-2)a}^*-\subscr{u}{min})\sum_{i=1}^{k-2}T_{i}^{i-1}\big)
        \end{align}
        
\noindent {\bf (ii) Case 2:}
        \begin{align*}
            &y_{k-2}+(v_{(k-2)a}^*-\subscr{u}{max})\sum_{i=1}^{k-2}T_{i}^{i-1}-\Delta x_{k-1}^{k-2}\subscr{u}{max}\\
            &\leq y_{k-1}\leq y_{k-2}+(v_{(k-2)a}^*-V)\sum_{i=1}^{k-2}T_{i}^{i-1}-\Delta x_{k-1}^{k-2}V
        \end{align*}
        and
       \begin{equation}\label{eq:n_coop_MD_2}
       \begin{split}
            &\Delta y_{k-1}^k<-\Delta x_{k}^{k-1}V+ (v_{(k-1)g}^*-V)\sum_{i=1}^{k}T_{i}^{i-1}+\\
            &U\big(\subscr{u}{max}\Delta x_{k-1}^{k-2}+\Delta y_{k-2}^{k-1}       -(v_{(k-2)a}^*-\subscr{u}{max})\sum_{i=1}^{k-2}T_{i}^{i-1}\big),
            \end{split}
        \end{equation}
        
%\end{itemize}}
where $v_{(k-2)a}^*$ determined by Algorithm \ref{algo:gesec}.
\end{Lemma}
\begin{proof}
Let us assume that this result holds for some $k = \bar{k}-1$. The idea is to prove this result using induction by deriving the conditions for cooperation between $E_{\bar k}$ and $E_{\bar{k}+1}$. For brevity, we will reuse Figure \ref{fig:MD_2_coop} with the two evaders $E_1$ and $E_2$ in the figure corresponding to $E_{\bar k}$ and $E_{\bar k+1}$ respectively.
Suppose that the initial location of $E_{\bar{k}}$ satisfies $y_{\bar k}\geq y_{\bar k-1}+(v_{(\bar k-1)a}^*-V)\sum_{i=1}^{\bar k-1}T_i^{i-1}-\Delta x_{\bar k}^{\bar k-1}V$ and the location of $E_{\bar{k}+1}$ satisfies $y_{\bar{k}+1}>y_{\bar k}-\Delta x_{\bar{k}+1}^{\bar{k}} V+ (v_{\bar{k}g}^*-V)\sum_{i=1}^{\bar{k}}T_{i}^{i-1}$. Then, from Lemma \ref{lem:greed_n_MD}, $E_{\bar k}$ moves with $\subscr{u}{max}$ to maximize the component $T_{\bar k}^{\bar k-1}$ out of its intercept time and $E_{\bar k+1}$ moves with $\subscr{u}{max}$ to maximize $T_{\bar k+1}^{\bar k}$. This implies that when $P$ completes stage 1 for the pursuit of $E_{\bar k}$, $P$ is below $E_{\bar k}$. As the pursuer follows the Manhattan pursuit strategy, the time taken to cover the path in the y-direction to intercept $E_{\bar k}$ and then $E_{\bar{k}+1}$ equals the time taken to cover the path in the y-direction to intercept only $E_{{\bar k}+1}$ from the pursuer's location (\ref{fig:MD_2_coop} (a)). Note that the pursuer is located at the intercept location of $E_{\bar k-1}$. %Mathematically,
% \begin{align*}
%     &T_{(k-1,k)}(\subscr{u}{max}) + T_{(k,k+1)}(\subscr{u}{max})\\
%     =&\frac{\Delta x_{(k-1,k)}+y_k-y_{k-1}+(v_k-v_{k-1})\sum_{i=1}^{k-1}T_{(i-1,i)}(v_1,\dots,v_i)}{1-\subscr{u}{max}} + \\
%     &\frac{\Delta x_{(k,k+1)}+y_{k+1}-y_k+(v_{k+1}-v_k)\sum_{i=1}^{k}T_{(i-1,i)}(v_1,\dots,v_i)}{1-\subscr{u}{max}}\\
%     =&\frac{\Delta x_{(k,k+1)}+\Delta x_{(k-1,k)}+y_{k+1}-y_{k-1}+(v_{k+1}-v_{k-1})\sum_{i=1}^{k-1}T_{(i-1,i)}(v_1,\dots,v_i)}{1-\subscr{u}{max}},
% \end{align*}
% where, $T_{(k-1,k)}(\subscr{u}{max})$ and $T_{(k,k+1)}(\subscr{u}{max})$ are time to intercept $E_k$ and $E_{k+1}$ when they both move with $\subscr{u}{max}$.
So, to increase the total time to intercept, $E_{\bar k}$ and $E_{\bar k+1}$ need to cooperate which can occur in only two ways.

The first is that $E_{\bar k}$ moves with speed $v_{\bar k}$ satisfying $\Delta x_{\bar k}^{\bar{k}-1}<(y_{\bar{k}-1}-y_{\bar k}+(v_{(\bar{k}-1)a}^*-v_{\bar{k}})\sum_{i=1}^{\bar{k}-1} T_{i}^{i-1})/v_{\bar{k}}$, which means that $P$ intercepts $E_{\bar k}$ below the intercept point of $E_{\bar k-1}$ and $E_{\bar k+1}$ moves greedily, i.e., with $\subscr{u}{max}$ (Fig. \ref{fig:MD_2_coop} (b)). This is possible only if
% \begin{align*}
    % &\Delta x_{(k-1,k)}<(y_{k-1}-y_k+(v_{k-1}-v_k)\sum_{i=1}^{k-1} T_{i-1,i})/\subscr{u}{min}\\
     $\Delta y_{\bar{k}-1}^{\bar{k}}\leq 
            (v_{(\bar{k}-1)a}^*-\subscr{u}{min})\sum_{i=1}^{\bar{k}-1}T_i^{i-1}-\Delta x_{\bar{k}}^{\bar{k}-1}\subscr{u}{min}$
%\end{align*}
holds.
If $\Delta y_{\bar{k}-1}^{\bar{k}}>
            (v_{(\bar{k}-1a)}^*-\subscr{u}{min})\sum_{i=1}^{\bar{k}-1}T_i^{i-1}-\Delta x_{\bar{k}}^{\bar{k}-1}\subscr{u}{min}$, then for any speed $v_{\bar k}$ for $E_{\bar k}$, the total distance covered in the $Y$-direction to intercept $E_{\bar k}$ and $E_{\bar k+1}$ will be the same as the total distance covered to intercept $E_{\bar k+1}$ from the pursuer's location, irrespective of $E_{\bar k}$'s choice.
            
The second case in which $E_{\bar k}$ and $E_{\bar k+1}$ cooperate is if $E_{\bar k}$ moves greedily with $\subscr{u}{max}$ and $E_{k+1}$ moves with speed $v_{\bar k+1}$ such that $y_{\bar{k}+1}<y_{\bar k}-\Delta x_{\bar{k}+1}^{\bar{k}} v_{\bar k+1}+ (v_{\bar{k}g}^*-v_{\bar k+1})\sum_{i=1}^{\bar{k}}T_{i}^{i-1}$, i.e., the intercept location of $E_{\bar k+1}$ is below the intercept location of $E_{\bar k}$ (Fig. \ref{fig:MD_2_coop} (c)). To determine which of the two scenarios yield greater intercept time, we arrive at the condition 
\begin{align*}
\sum_{i=1}^k T_{i}^{i-1}+T_{\bar k+1}^{\bar k}(\subscr{u}{max}) > \sum_{i=1}^k T_{i}^{i-1}+ T_{\bar k+1}^{\bar k}(\subscr{u}{min})
\end{align*}
which yields the conclusion that $E_{\bar k}$ should cooperate with $E_{\bar k+1}$ only if equation \eqref{eq:n_coop_MD} holds. This concludes our proof for case 1. Case 2 can be proved by following the steps for Case 1 and has been omitted for brevity.
\end{proof}
\medskip

Lemma \ref{lem:coop_nE_MD} establishes the conditions for cooperation between any two consecutive evaders. The next result characterizes the cooperative strategies of the evaders.

\begin{theorem}[Cooperative strategy for $E_{k-1}$]\label{thm:opt_strat_n_MD}
% Given the initial locations $(x_{k-1},y_{k-1})$, $(x_k,y_k)$, and $(X,Y)$ of $E_{k-1}$, $E_k$ and $P$, respectively, 
If the conditions on cooperation in Lemma \ref{lem:coop_nE_MD} hold, then the optimal strategy $v_{(k-1)c}^*$ for $E_{k-1}$ during cooperation with $E_k$ is\\

$v_{(k-1)c}^*=
\begin{cases}
\subscr{u}{min}, \text{ for case 1 of Lemma \ref{lem:coop_nE_MD} },\\
\subscr{u}{max}, \text{ for case 2 of Lemma \ref{lem:coop_nE_MD}. }
\end{cases}$

% \begin{itemize}
%     \item For case 1 of Lemma \ref{lem:coop_nE_MD}: 
        % \begin{align*}
        %     &y_{n-2}+(v_{(n-2)}^*-V)\sum_{i=1}^{n-2}T_{(i-1,i)}-\Delta x_{(n-2,n-1)}V\leq y_{n-1}\leq y_{n-2}\\
        %     &+(v_{(n-2)}^*-\subscr{u}{min})\sum_{i=1}^{n-2}T_{(i-1,i)}-\Delta x_{(n-2,n-1)}\subscr{u}{min}
        % \end{align*}
        % and 
        % \begin{align*}
        %     &y_n>y_{n-1}-\Delta x_{(n-1,n)}V+ (v_{(n-1)g}^*-V)\sum_{i=1}^{n}T_{(i-1,i)} + \\
        %     &2\big(\frac{\subscr{u}{min}\Delta x_{(n-2,n-1)}+y_{n-1}-y_{n-2}-(v_{n-2}^*-\subscr{u}{min})\sum_{i=1}^{n-2}T_{(i-1,i)}}{2+\subscr{u}{min}-\subscr{u}{max}}\big)
        % \end{align*}
        % then,
        % \[
        %     v_{(k-1)}^* = \subscr{u}{min} \text{ or },
        % \]
        % \[
        %     v_n^*=\subscr{u}{max}
        % \]
        % \qquad else
        % \[
        %     v_{n-1}^* = \subscr{u}{max}
        % \]
        % \[
        %     v_n^*=\subscr{u}{min}
        % \]
   % \item For case 2 of Lemma \ref{lem:coop_nE_MD}: 
    %     \begin{align*}
    %         &y_{n-2}+(v_{(n-2)}^*-\subscr{u}{max})\sum_{i=1}^{n-2}T_{(i-1,i)}-\Delta x_{(n-2,n-1)}\subscr{u}{max}\leq y_{n-1}\leq y_{n-2}+\\
    %         &(v_{(n-2)}^*-V)\sum_{i=1}^{n-2}T_{(i-1,i)}-\Delta x_{(n-2,n-1)}V
    %     \end{align*}
    %     and
    %   \begin{align*}
    %         &y_n<y_{n-1}-\Delta x_{(n-1,n)}V+ (v_{(n-1)g}^*-V)\sum_{i=1}^{n}T_{(i-1,i)} + \\
    %         &2\big(\frac{\subscr{u}{max}\Delta x_{(n-2,n-1)}+y_{n-1}-y_{n-2}-(v_{n-2}^*-\subscr{u}{max})\sum_{i=1}^{n-2}T_{(i-1,i)}}{2+\subscr{u}{min}-\subscr{u}{max}}\big)
    %     \end{align*}
    %     then,
        %      \[
        %     v_{(k-1)}^* = \subscr{u}{max}.
        % \]
        % \[
        %     v_n^*=\subscr{u}{min}
        % \]
        % \qquad else
        %     \[
        %     v_{n-1}^* = \subscr{u}{min}
        % \]
        % \[
        %     v_n^*=\subscr{u}{max}
        % \]
%\end{itemize}
\end{theorem}
\begin{proof}
% From the 2 evader case we know that if
% \begin{align*}
%             Y-\frac{\Delta x_{(0,1)}(\subscr{u}{min}+\subscr{u}{max})}{2+\subscr{u}{min}-\subscr{u}{max}}\leq y_1\leq Y-\Delta x_{(0,1)}\subscr{u}{min}
% \end{align*}
%         and 
% \begin{align*}
%             y_2>y_1-\Delta x_{(1,2)}\frac{(\subscr{u}{min}+\subscr{u}{max})}{2+\subscr{u}{min}-\subscr{u}{max}}+ \big(v_{1g}^*-\frac{(\subscr{u}{min}+\subscr{u}{max})}{2+\subscr{u}{min}-\subscr{u}{max}}\big)T_{(0,1)} + 2(\frac{\subscr{u}{min}\Delta x_{(1,2)}+y_1-Y}{2+\subscr{u}{min}-\subscr{u}{max}}).
% \end{align*}
% then, $v_1^*=\subscr{u}{min}$ and $v_2^*=\subscr{u}{max}$.
Assume that the result holds for $k=\bar k-1$. The idea is to prove this result by induction by deriving this result for $E_{\bar k}$. Suppose case 1 from Lemma \ref{lem:coop_nE_MD} holds for $E_{\bar k}$ and $E_{\bar k+1}$. From Lemma \ref{lem:coop_nE_MD}, we know in order to cooperate with $E_{\bar k+1}$, $E_{\bar k}$ moves with a speed $v_{\bar k}$ such that $\Delta x_{\bar k}^{\bar{k}-1}<(y_{\bar{k}-1}-y_{\bar k}+(v_{\bar{k}-1}^*-v_{\bar{k}})\sum_{i=1}^{\bar{k}-1} T_{i}^{i-1})/v_{\bar{k}}$, i.e., $E_{\bar k}$ is intercepted below the pursuer's location. We also know from the same lemma that these conditions on cooperation ensure that the total time to intercept while cooperation is higher than total time to intercept when the evaders move greedy. Now, since $T_{\bar k}^{\bar k-1}$ is monotonic in $v_{\bar k}$, $v_{\bar kc}^*=\subscr{u}{min}$.
Similar steps can be followed for case 2.
\end{proof}

\begin{remark}[Sandwiched evader]\label{lem:sandwhich}
For some $i \in \{1,\dots, n\}$, if Lemma \ref{lem:coop_nE_MD} holds for evader $E_{i-1}$ and $E_i$ as well as $E_i$ and $E_{i+1}$, then evader $E_i$ moves greedy.
\end{remark}

\section{Fundamental Limit}\label{sec:fundamental}
In the previous sections, we considered that the pursuer followed a fixed strategy to capture all evaders. We now establish a fundamental upper bound, for \emph{a large number of evaders}, on the total time taken to intercept all evaders by the pursuer following \emph{any} strategy. We first provide some existing results that will be useful in establishing the bound.

Given a set of $m$ points, a \emph{Euclidean minimum Hamiltonian path} (EMHP) is the shortest path through $m$ points such that each point is visited exactly once. When the points are translating with some constant speed $v\in (0,1)$, then the shortest tour though the points is called \emph{Translational minimum Hamiltonian path} (TMHP) \cite{hammar1999approximation}.

\begin{Lemma}[Length of EMHP tour]\label{lem:Fews}
Given $m$ points in a $l\times h$ rectangle in the plane, where $h\in \mathbb{R}_{>0}$ and $l\in \mathbb{R}_{>0}$, there exists a path that starts
from a unit length edge of the rectangle, passes through each of
the $m$ points exactly once, and terminates on the opposite unit
length edge, with length upper bounded by $\sqrt{2lhm}+h+2.5$
\end{Lemma}
\begin{proof}
The proof is similar to the proof provided in~\cite{bopardikar2010dynamic} for a $1\times h$ rectangle and thus, has been omitted. 
\end{proof}
\medskip
To calculate the EMHP tour through translating points $s, s_1,\dots, s_f, f$ that move with speed $v$, the points are scaled by defining a conversion map $C_v:\mathbb{R}^2\to \mathbb{R}^2$ such that $C_v(x,y)=(\frac{x}{\sqrt{(1-v^2)}},\frac{y}{1-v^2})$ \cite{hammar1999approximation}.
\begin{Lemma}[Length of TMHP tour~\cite{hammar1999approximation}]\label{lem:TMHP_tour}
Let the initial and final point be denoted as $s=(x_s,y_s)$ and $f=(x_f,y_f)$ respectively, and $v \in (0,1)$ denote a constant speed of all evaders, then the length of the TMHP tour is $\frac{v(y_f-y_s)}{1-v^2} + \mathcal{L}_E(C_v(s),{C_v(s_1),\dots, C_v(s_f)},C_v(f))$
where, $\mathcal{L}_E(C_v(s),{C_v(s_1),\dots, C_v(s_f)},C_v(f))$ denotes the length of the EMHP starting with point $s$, moving through points ${s_1,\dots, s_f}$ and ending at point $f$.
\end{Lemma}
The optimal order followed by the vehicle in the TMHP solution is the same as the optimal order followed by the vehicle in the EMHP solution.\\

Denote $\subscr{n}{max}\in \mathbb{Z}_0^+$ as the total number of evaders that move with $\subscr{u}{max}$ and $\subscr{n}{min}=n-\subscr{n}{max}$ as the total number of evaders that move with $\subscr{u}{min}$. Let $\mathcal{\subscr{A}{max}}$ and $\mathcal{\subscr{A}{min}}$ denote the area of the smallest enclosing rectangular environment that the $\subscr{n}{max}$ and $\subscr{n}{min}$ evaders occupy initially. We assume that all of the evaders are initially located within a rectangular environment of area $\mathcal{A}$. The pursuer's strategy is to capture all the $\subscr{n}{max}$ evaders first, followed by capturing all the evaders moving with $\subscr{u}{min}$. This is because if the pursuer captures the $\subscr{n}{min}$ evaders first then naturally, the evaders moving with $\subscr{u}{max}$ will be further away from the pursuer.

Let $T_{\subscr{n}{max}}$ be the time taken by the vehicle to capture all of the $\subscr{n}{max}$ evaders and $T_{\subscr{n}{min}}^{\subscr{n}{max}}$ be the time taken to intercept the last evader that moves with $\subscr{u}{max}$ and the first evader that moves with $\subscr{u}{min}$ after capturing all of the $\subscr{n}{max}$ evaders respectively. Let $T_{\subscr{n}{min}}$ be the total time taken by the vehicle to capture all of the remaining $\subscr{n}{min}-1$ evaders. The next result characterizes an upper bound on the time taken by the pursuer to capture all evaders following any strategy.

\begin{theorem}[Upper bound on intercept time]\label{thm:upper_bound}
Let $\Delta y$ and $\Delta x$ be the difference between the initial $y$ and $x$-coordinate of the last evader captured moving with $\subscr{u}{max}$ and the first evader that is captured moving with $\subscr{u}{min}$. Then, from Lemma \ref{lem:Fews} and Lemma \ref{lem:TMHP_tour},
the total time taken by the pursuer to capture all evaders is $T = T_{\subscr{n}{max}} + T_{\subscr{n}{min}}^{\subscr{n}{max}} + T_{\subscr{n}{min}}$
where,
\begin{align*}
    T_{\subscr{n}{max}} = \sqrt{\frac{2\mathcal{\subscr{A}{max}}\subscr{n}{max}}{(1-\subscr{u}{max}^2)^{3/2}}},\text{ } T_{\subscr{n}{min}} = \sqrt{\frac{2\mathcal{\subscr{A}{min}}(\subscr{n}{min}-1)}{(1-\subscr{u}{min}^2)^{3/2}}}\\
    T_{\subscr{n}{min}}^{\subscr{n}{max}} = \frac{\subscr{u}{min}}{1-\subscr{u}{min}^2}(\Delta y + (\subscr{u}{min}-\subscr{u}{max})T_{\subscr{n}{max}}) + \\
    \sqrt{\frac{\Delta x^2}{1-\subscr{u}{min}^2}+\frac{(\Delta y + (\subscr{u}{min}-\subscr{u}{max})T_{\subscr{n}{max}})^2}{(1-\subscr{u}{min}^2)^2}}.
\end{align*}
Moreover, for large $n$,  $T$ is maximum for
\begin{align*}
     \subscr{n}{max}^* =\bigg\lfloor\frac{(\subscr{u}{min}-\subscr{u}{max})^2n}{(1-\subscr{u}{min}^2)^{\frac{1}{2}}(1-\subscr{u}{max}^2)^{\frac{3}{2}}+(\subscr{u}{min}-\subscr{u}{max})^2}\bigg\rceil,
\end{align*}
where $\lfloor x \rceil$ denotes the integer nearest to $x$.
\end{theorem}
\begin{proof}
The outline of the proof is as follows. The expression for $T_{\subscr{n}{max}}$ and $T_{\subscr{n}{min}}$ follows directly from Lemma \ref{lem:Fews} and noting that $n$ is large. The expression for $T_{\subscr{n}{min}}^{\subscr{n}{max}}$ follows from \cite{hammar1999approximation}. Consider that the vehicle has just finished capturing all $\subscr{n}{max}$ evaders. Then, all the evaders moving with $\subscr{u}{min}$ would have translated $\subscr{u}{min}T_{\subscr{n}{max}}$ in the $y$ direction. Note that the area $\mathcal{\subscr{A}{min}}$ will remain the same as it was initially. Since $T_{\subscr{n}{max}}$ is large for large $\subscr{n}{max}$, the distance between the vehicle after capturing the last evader moving with $\subscr{u}{max}$ and the first evader moving with $\subscr{u}{min}$ will be large. Furthermore, since $\subscr{u}{max}>\subscr{u}{min}$ and $T_{\subscr{n}{max}}$ is large, the pursuer will always be above all evaders moving with $\subscr{u}{min}$ after capturing $\subscr{n}{max}$ evaders and thus, we get the expression for $T_{\subscr{n}{min}}^{\subscr{n}{max}}$. The evaders can select $\subscr{n}{max}$ such that the total time $T$ is maximized. Mathematically, $\subscr{n}{max}^* = \operatorname*{arg\,max}_{\subscr{n}{max}}  T(\subscr{n}{max})$.
% \begin{align*}
%     \subscr{n}{max}^* = \operatorname*{arg\,max}_{\subscr{n}{max}}  T(\subscr{n}{max})
% \end{align*}
If we relax the requirement of $\subscr{n}{max}$ to be a real number then the function $T:\mathbb{R}\to \mathbb{R}$ is concave with global maximum in the domain $[0,n]$. This follows as $\frac{dT^2}{d\subscr{n}{max}^2}<0$. Thus, to find the maximizer, we use the first derivative test, i.e., $\frac{dT}{d\subscr{n}{max}}=0$ and then find the closest integer value that maximizes $T$.
By taking the derivative of $T$ with respect to $\subscr{n}{max}$, we get $\frac{dT_{\subscr{n}{max}}}{d\subscr{n}{max}}+\frac{dT_{\subscr{n}{min}}^{\subscr{n}{max}}}{d\subscr{n}{max}}+\frac{dT_{\subscr{n}{min}}}{d\subscr{n}{max}}$, where,
\begin{align*}
   \frac{dT_{\subscr{n}{max}}}{d\subscr{n}{max}}= \frac{ \sqrt{\mathcal{\subscr{A}{max}}}}{(1-\subscr{u}{max}^2)^{3/4}\sqrt{2\subscr{n}{max}}},\\
   \frac{dT_{\subscr{n}{min}}}{d\subscr{n}{max}} = -\frac{ \sqrt{\mathcal{\subscr{A}{min}}}}{(1-\subscr{u}{min}^2)^{3/4}\sqrt{2(n-\subscr{n}{max}-1)}},\\
    \frac{dT_{\subscr{n}{min}}^{\subscr{n}{max}}}{d\subscr{n}{max}} = \frac{\subscr{u}{min}}{1-\subscr{u}{min}^2}(\subscr{u}{min}-\subscr{u}{max})\frac{dT_{\subscr{n}{max}}}{d\subscr{n}{max}}+\\
   \frac{(\Delta y+(\subscr{u}{min}-\subscr{u}{max})T_{\subscr{n}{max}})(\subscr{u}{min}-\subscr{u}{max})\frac{dT_{\subscr{n}{max}}}{d\subscr{n}{max}}}{(1-\subscr{u}{min}^2)^2\sqrt{\frac{\Delta x^2}{1-\subscr{u}{min}^2}+\frac{(\Delta y+(\subscr{u}{min}-\subscr{u}{max})T_{\subscr{n}{max}})^2}{(1-\subscr{u}{min}^2)^2}}}.
\end{align*}
In the worst case, the evader initial locations will cover the area $\mathcal{A}$ completely, i.e., $\mathcal{\subscr{A}{max}}=\mathcal{\subscr{A}{min}}=\mathcal{A}$.
Adding and equating to zero, and noting that $T_{\subscr{n}{max}}$ is large, $\Delta y\ll(\subscr{u}{min}-\subscr{u}{max})T_{\subscr{n}{max}}$ and $1-\subscr{u}{min}\subscr{u}{max}\ll\frac{\Delta y+(\subscr{u}{min}-\subscr{u}{max})T_{\subscr{n}{max}}}{(1-\subscr{u}{min}^2)\sqrt{\frac{\Delta x^2}{1-\subscr{u}{min}^2}+\frac{(\Delta y^2+(\subscr{u}{min}-\subscr{u}{max})T_{\subscr{n}{max}})^2}{(1-\subscr{u}{min}^2)^2}}}$
yields
\begin{align*}
   \subscr{n}{max}^*= \frac{2\Delta u^4\mathcal{A}n-(1-\subscr{u}{min}^2)^{3/2}\Delta x^2(1-\subscr{u}{max}^2)^3}{2\Delta u^2\mathcal{A}((1-\subscr{u}{min}^2)^{\frac{1}{2}}(1-\subscr{u}{max}^2)^{\frac{3}{2}}+\Delta u^2)},
\end{align*}
where $\Delta u=\subscr{u}{min}-\subscr{u}{max}$. Since, $2\Delta u^4\mathcal{A}n\gg(1-\subscr{u}{min}^2)^{\frac{3}{2}}\Delta x^2(1-\subscr{u}{max})^3$, we get the result. Furthermore, since $\frac{d^2T}{d\subscr{n}{max}^2}<0$ at the critical point $\subscr{n}{max}^*$, $\subscr{n}{max}^*$ is indeed the point of maximum. This concludes the proof.
\end{proof}

\section{Simulation Results}\label{sec:Sims}
We first present the numerical results for Algorithm \ref{algo:gesec}. 
We compare the mean of the total time to intercept all evaders using Algorithm \ref{algo:gesec} to the mean of the total time to intercept all evaders by randomly sampling over the evader speeds of either $\subscr{u}{min}$ or $\subscr{u}{max}$ (see Figure~\ref{fig:sims}). For each value of $n$, we randomly generate the initial locations of the evaders and the pursuer and we consider 50 Monte Carlo trials. To select the best evader speeds, we choose $10nln(2/\delta)$ samples uniformly randomly over the set, which guarantees that the violation probability is less than a small quantity $\delta$~\cite{alamo2010sample}, where $\delta = 0.1$. We compute the maximum over the samples and then report the mean value in Figure~\ref{fig:sims}. We observe that Algorithm \ref{algo:gesec} outperforms random sampling.\\ 

Figure \ref{fig:upperbound} shows a comparison when $\subscr{n}{max}$ is selected uniformly randomly to the upper bound obtained by $\subscr{n}{max}^*$ for given initial locations. To obtain the EMHP tour required for the time to intercept evaders, the {\ttfamily linkern}\footnote{The TSP solver {\ttfamily linkern} is
  freely available for academic research use at {\ttfamily
    http://www.math.uwaterloo.ca/tsp/concorde/}.} solver was used. We consider 50 Monte Carlo trials for each value of $n$ and report the mean and standard deviation. It is observed that the total time to intercept all evaders by randomly selecting $\subscr{n}{max}$ is well below the upper bound obtained from $\subscr{n}{max}^*$. Thus, by performing an additional optimization to select $\subscr{n}{max}$ the evaders can reach the upper bound on time to intercept. This means that a strategy that only depends on $\subscr{n}{max}$ may be sub-optimal for the evaders.
 \begin{figure}[t]
    \centering
    \includegraphics[scale=0.55]{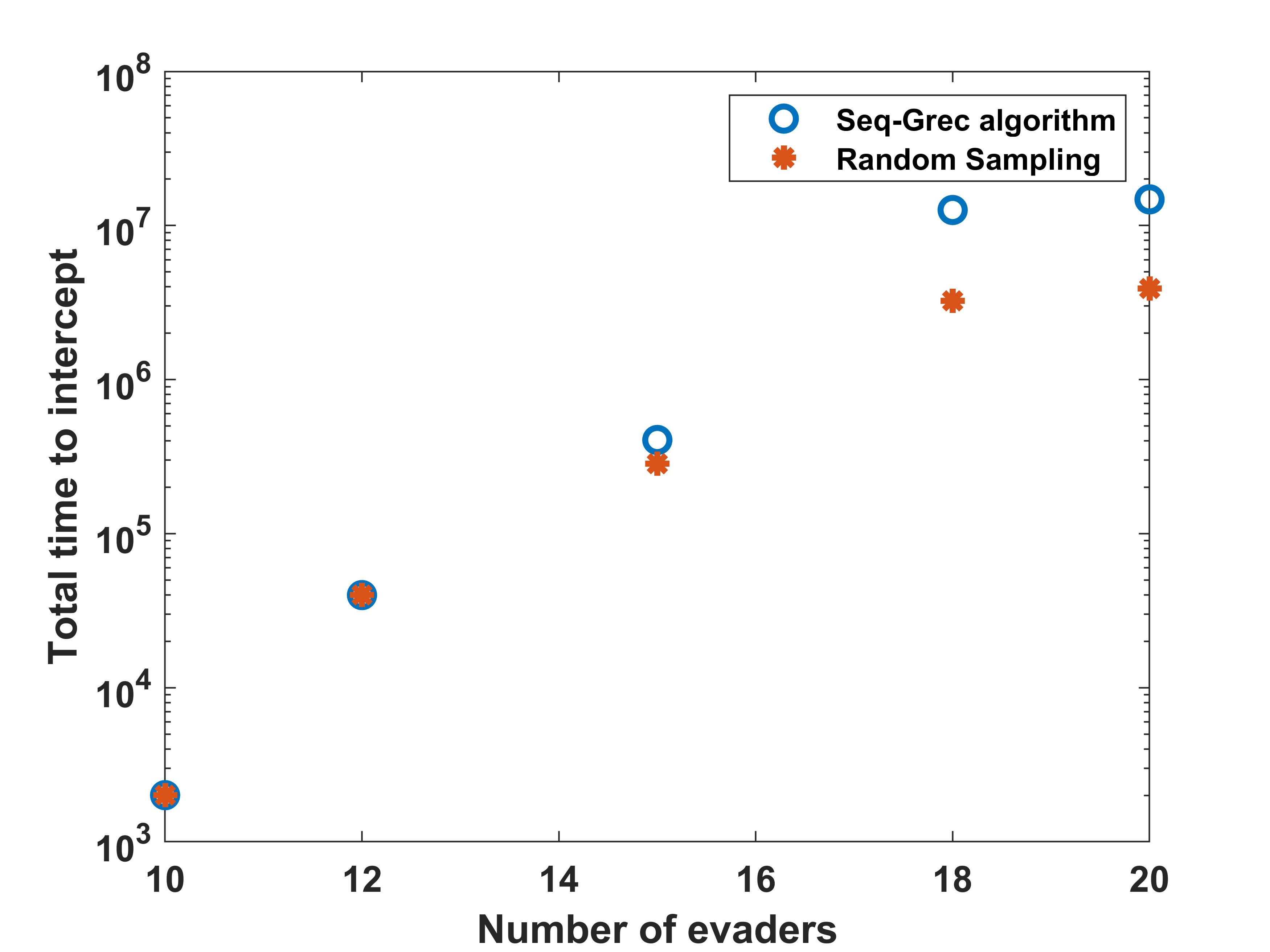}
    \caption{\small{Comparison of Seq-Grec with Random Sampling. The blue circles represent the mean of the total time to intercept of Seq-Grec Algorithm. The orange star represents the mean over the samples of Random Sampling}}
    \label{fig:sims}
\end{figure}
 \begin{figure}[t]
    \centering
    \includegraphics[scale=0.55]{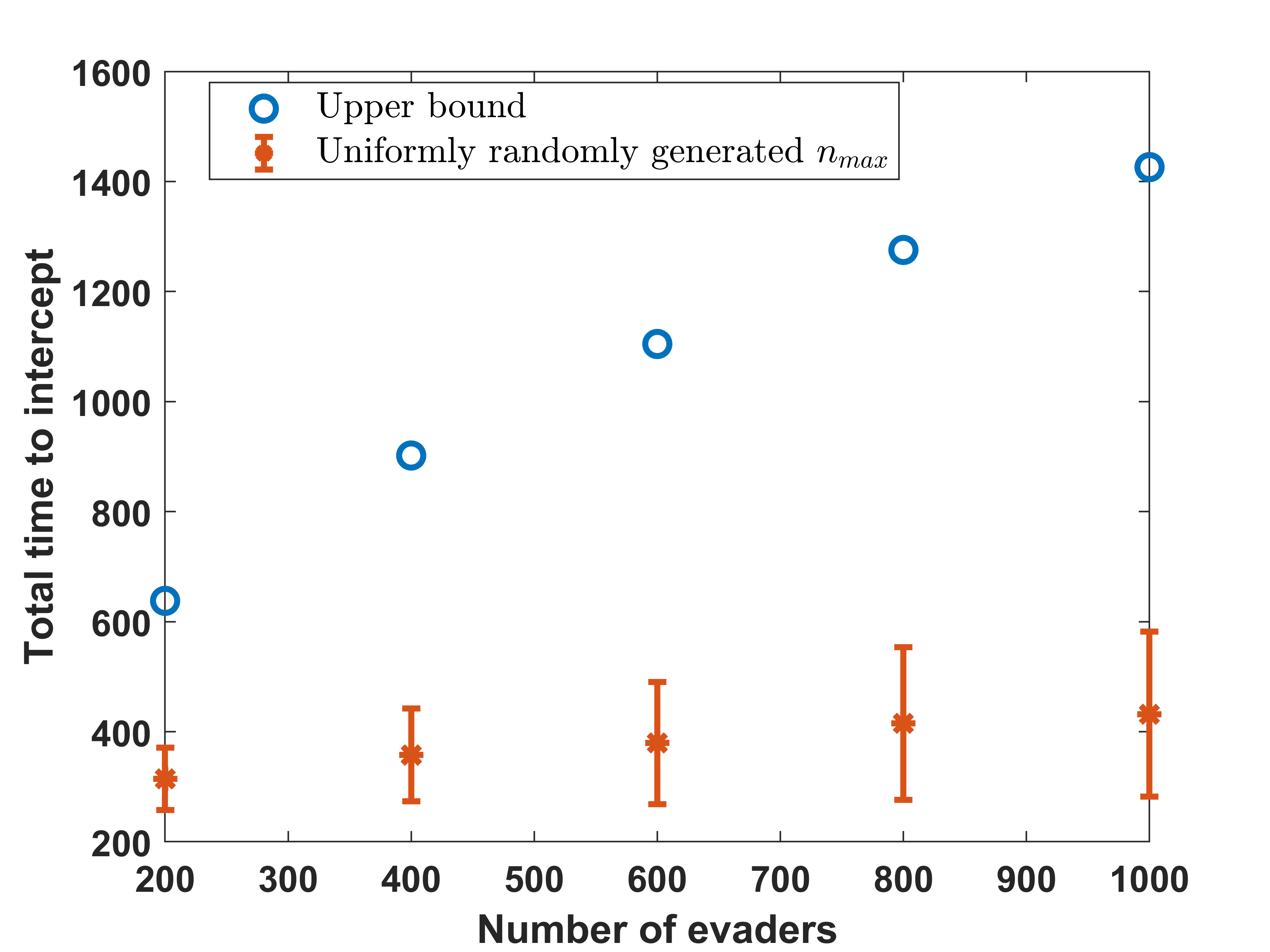}
    \caption{\small{Comparison of the total time to intercept when $\subscr{n}{max}$ is randomly uniformly selected to that of the upper bound obtained from $\subscr{n}{max}^*$.}}
    \label{fig:upperbound}
\end{figure}
\section{Conclusions and Future Work}\label{sec:Conc}
An optimal evasion problem between single a pursuer and multiple evaders was addressed. It is shown that by enforcing cooperation among evaders, they are able to maximize the total interception time. Conditions where cooperation is optimal are also presented which are crucial to implement the cooperative strategies. An upper bound on the total time to intercept all evaders is also presented.

In subsequent work, a generalized setup of multiple pursuers and evaders will be considered. Constant factor approximations for both, the evaders and the pursuers will also be addressed. Identifying which evaders should move with $\subscr{u}{max}$ is another possible extension.

\bibliographystyle{IEEEtran}
\bibliography{main}

\begin{thebibliography}{10}
\providecommand{\url}[1]{#1}
\csname url@samestyle\endcsname
\providecommand{\newblock}{\relax}
\providecommand{\bibinfo}[2]{#2}
\providecommand{\BIBentrySTDinterwordspacing}{\spaceskip=0pt\relax}
\providecommand{\BIBentryALTinterwordstretchfactor}{4}
\providecommand{\BIBentryALTinterwordspacing}{\spaceskip=\fontdimen2\font plus
\BIBentryALTinterwordstretchfactor\fontdimen3\font minus
  \fontdimen4\font\relax}
\providecommand{\BIBforeignlanguage}[2]{{%
\expandafter\ifx\csname l@#1\endcsname\relax
\typeout{** WARNING: IEEEtran.bst: No hyphenation pattern has been}%
\typeout{** loaded for the language `#1'. Using the pattern for}%
\typeout{** the default language instead.}%
\else
\language=\csname l@#1\endcsname
\fi
#2}}
\providecommand{\BIBdecl}{\relax}
\BIBdecl

\bibitem{ragesh2014analysis}
R.~Ragesh, A.~Ratnoo, and D.~Ghose, ``Analysis of evader survivability
  enhancement by decoy deployment,'' in \emph{2014 American Control
  Conference}.\hskip 1em plus 0.5em minus 0.4em\relax IEEE, 2014, pp.
  4735--4740.

\bibitem{isaacs1999differential}
R.~Isaacs, \emph{Differential games: a mathematical theory with applications to
  warfare and pursuit, control and optimization}.\hskip 1em plus 0.5em minus
  0.4em\relax Courier Corporation, 1999.

\bibitem{makkapati2019optimal}
V.~R. Makkapati and P.~Tsiotras, ``Optimal evading strategies and task
  allocation in multi-player pursuit--evasion problems,'' \emph{Dynamic Games
  and Applications}, vol.~9, no.~4, pp. 1168--1187, 2019.

\bibitem{girard2015proportional}
A.~R. Girard and P.~T. Kabamba, ``Proportional navigation: optimal homing and
  optimal evasion,'' \emph{SIAM Review}, vol.~57, no.~4, pp. 611--624, 2015.

\bibitem{selvakumar2016evasion}
J.~Selvakumar and E.~Bakolas, ``Evasion from a group of pursuers with a
  prescribed target set for the evader,'' in \emph{2016 American Control
  Conference (ACC)}.\hskip 1em plus 0.5em minus 0.4em\relax IEEE, 2016, pp.
  155--160.

\bibitem{fuchs2010cooperative}
Z.~E. Fuchs, P.~P. Khargonekar, and J.~Evers, ``Cooperative defense within a
  single-pursuer, two-evader pursuit evasion differential game,'' in \emph{49th
  IEEE Conference on Decision and Control (CDC)}.\hskip 1em plus 0.5em minus
  0.4em\relax IEEE, 2010, pp. 3091--3097.

\bibitem{zemskov1997construction}
K.~Zemskov and A.~Pashkow, ``Construction of optimal position strategies in a
  differential pursuit-evasion game with one pursuer and two evaders,''
  \emph{Journal of applied mathematics and mechanics}, vol.~61, no.~3, pp.
  391--399, 1997.

\bibitem{oyler2016pursuit}
D.~W. Oyler, P.~T. Kabamba, and A.~R. Girard, ``Pursuit--evasion games in the
  presence of obstacles,'' \emph{Automatica}, vol.~65, pp. 1--11, 2016.

\bibitem{garcia2014cooperative}
E.~Garcia, D.~W. Casbeer, K.~Pham, and M.~Pachter, ``Cooperative aircraft
  defense from an attacking missile,'' in \emph{53rd IEEE Conference on
  Decision and Control}.\hskip 1em plus 0.5em minus 0.4em\relax IEEE, 2014, pp.
  2926--2931.

\bibitem{chikrii1987pursuit}
A.~A. Chikrii and S.~Kalashnikova, ``Pursuit of a group of evaders by a single
  controlled object,'' \emph{Cybernetics and Systems Analysis}, vol.~23, no.~4,
  pp. 437--445, 1987.

\bibitem{shevchenko2008guaranteed}
I.~Shevchenko, ``Guaranteed approach with the farthest of the runaways,''
  \emph{Automation and Remote Control}, vol.~69, no.~5, pp. 828--844, 2008.

\bibitem{belousov2010solving}
A.~Belousov, Y.~I. Berdyshev, A.~Chentsov, and A.~Chikrii, ``Solving the
  dynamic traveling salesman game problem,'' \emph{Cybernetics and Systems
  Analysis}, vol.~46, no.~5, pp. 718--723, 2010.

\bibitem{liu2013evasion}
S.-Y. Liu, Z.~Zhou, C.~Tomlin, and K.~Hedrick, ``Evasion as a team against a
  faster pursuer,'' in \emph{2013 American Control Conference}.\hskip 1em plus
  0.5em minus 0.4em\relax IEEE, 2013, pp. 5368--5373.

\bibitem{Scott2018OptimalES}
W.~L. Scott and N.~E. Leonard, ``Optimal evasive strategies for multiple
  interacting agents with motion constraints,'' \emph{Automatica}, vol.~94, pp.
  26--34, 2018.

\bibitem{krishnamoorthy2013optimal}
K.~Krishnamoorthy, S.~Darbha, P.~P. Khargonekar, D.~Casbeer, P.~Chandler, and
  M.~Pachter, ``Optimal minimax pursuit evasion on a \text{Manhattan} grid,''
  in \emph{2013 American Control Conference}.\hskip 1em plus 0.5em minus
  0.4em\relax IEEE, 2013, pp. 3421--3428.

\bibitem{hammar1999approximation}
M.~Hammar and B.~J. Nilsson, ``Approximation results for kinetic variants of
  {TSP},'' in \emph{International Colloquium on Automata, Languages, and
  Programming}.\hskip 1em plus 0.5em minus 0.4em\relax Springer, 1999, pp.
  392--401.

\bibitem{bopardikar2010dynamic}
S.~D. Bopardikar, S.~L. Smith, F.~Bullo, and J.~P. Hespanha, ``Dynamic vehicle
  routing for translating demands: Stability analysis and receding-horizon
  policies,'' \emph{IEEE Transactions on Automatic Control}, vol.~55, no.~11,
  pp. 2554--2569, 2010.

\bibitem{alamo2010sample}
T.~Alamo, R.~Tempo, and A.~Luque, ``On the sample complexity of randomized
  approaches to the analysis and design under uncertainty,'' in
  \emph{Proceedings of the 2010 American Control Conference}.\hskip 1em plus
  0.5em minus 0.4em\relax IEEE, 2010, pp. 4671--4676.

\end{thebibliography}
\end{document}